\documentclass[runningheads,proof]{llncs}
\usepackage[left]{lineno}

\usepackage{hyperref}

\usepackage{amsmath,amssymb,latexsym}
\usepackage{paralist,array,url}
\usepackage[ruled,lined,linesnumbered,boxed,vlined]{algorithm2e}
\usepackage{afterpage}
\usepackage{pgfplots}

\usetikzlibrary{patterns}
\usepgfplotslibrary{groupplots}

\newcommand{\set}[1]{\left\{#1\right\}}

\title{Pattern Backtracking Algorithm\\
for the Workflow Satisfiability Problem\\
with User-Independent Constraints}
\titlerunning{Pattern Backtracking Algorithm for the Workflow Satisfiability Problem}
\author{
D. Karapetyan\inst{1}
\and A. Gagarin\inst{2}
\and G. Gutin\inst{2}
}
\institute{
University of Nottingham, UK\\ 
\email{Daniel.Karapetyan@gmail.com}\\
           \and
Royal Holloway, University of London, UK \email{\{Andrei.Gagarin,G.Gutin\}@rhul.ac.uk}\\
           }
           
\usepackage{color}

\makeatletter
\g@addto@macro{\UrlBreaks}{\UrlOrds}
\makeatother

\begin{document}

\maketitle
\begin{abstract}
The workflow satisfiability problem (WSP) asks whether there exists an assignment of authorised users to the steps in a workflow specification, subject to certain constraints on the assignment.
(Such an assignment is called valid.)
The problem is NP-hard even when restricted to the large class of user-independent constraints.
Since the number of steps $k$ is relatively small in practice, it is natural to consider a parametrisation of the WSP by $k$.
We propose a new fixed-parameter algorithm to solve the WSP with user-independent constraints.
The assignments in our method are partitioned into equivalence classes such that the number of classes is exponential in $k$ only.
We show that one can decide, in polynomial time, whether there is a valid assignment in an equivalence class.
By exploiting this property, our algorithm reduces the search space to the space of equivalence classes, which it browses within a backtracking framework, hence emerging as an efficient yet relatively simple-to-implement or generalise solution method.
%The algorithm is based on the idea of partitioning all assignments into equivalence classes and on the backtracking framework resulting in an efficient yet relatively simple-to-implement or generalise algorithm.
We empirically evaluate our algorithm against the state-of-the-art methods and show that it clearly wins the competition on the whole range of our test problems and significantly extends the domain of practically solvable instances of the WSP.
\end{abstract}

\section{Introduction}\label{sec:intro}
% problem in general; general motivation
In the \emph{workflow satisfiability problem} (WSP), we aim at assigning authorised users to the steps in a workflow specification, subject to some constraints arising from business rules and practices.
The WSP has applications in information access control (e.g.\ see~\cite{ANSI04,BaBuKa14,BeFeAt99}), and it is extensively studied in the security research community~\cite{BaBuKa14,BeFeAt99,Cr05,WaLi10}.  
In the WSP, we are given a set $U$ of \emph{users}, a set $S$ of \emph{steps}, a set $\mathcal{A} = \set{A(u) : u \in U}$ of \emph{authorisation lists}, where $A(u) \subseteq S$ denotes the set of steps for which user $u$ is authorised, and a set $C$ of \emph{(workflow) constraints}.
In general, a \emph{constraint} $c \in C$ can be described as a pair $c = (T, \Theta)$, where $T \subseteq S$ is the \emph{scope} of the constraint and $\Theta$ is a set of functions from $T$ to $U$ which specifies those assignments of steps in $T$ to users in $U$ that satisfy the constraint (authorisations disregarded).  Authorisations and constraints described in WSP literature are relatively simple such that we may assume that all authorisations and constraints can be checked in polynomial time (in $|U|$, $|S|$ and $|C|$).

Given a \emph{workflow} $W = (S, U, \mathcal{A}, C)$, $W$ is \emph{satisfiable} if there exists a function $\pi: S \rightarrow U$ such that
\begin{itemize}
\item for all $s \in S$, $s \in A(\pi(s))$ (each step is allocated to an authorised user);
\item for all $(T,\Theta) \in C$, $\pi|_{T}  \in \Theta$ (every constraint is satisfied).
\end{itemize}
A function $\pi: S \rightarrow U$ is an \emph{authorised} (\emph{eligible}, \emph{valid}, respectively) \emph{complete plan} if it satisfies the first condition above (the second condition, both conditions, respectively).

For example, consider the following instance of WSP\@.
The step and user sets are $S = \{ s_1, s_2, s_3, s_4 \}$ and $U = \{u_1,u_2, \dots, u_5\}$. 
The authorisation lists are $A(u_1) = \{ s_1, s_2, s_3, s_4 \}$, $A(u_2) = \{ s_1 \}$, $A(u_3) = \{ s_2 \}$, $A(u_4) = A(u_5) = \{ s_3, s_4 \}$. 
The constraints are $(s_1,s_2,=)$ (the same user must be assigned to $s_1$ and $s_2$), $(s_2, s_3,\neq)$ ($s_2$ and $s_3$ must be assigned to different users),  $(s_3, s_4,\neq)$, and $(s_4,s_1,\neq)$. 
Since the function $\pi$ assigning $u_1$ to $s_1$ and $s_2$, $u_4$ to $s_3$, and $u_5$ to $s_4$ is a valid complete plan, the workflow is satisfiable.  

Clearly, not every workflow is satisfiable, and hence
it is important to be able to determine whether a workflow is satisfiable or not and, if it is satisfiable, to find a valid complete plan.
Unfortunately, the WSP is  NP-hard~\cite{WaLi10} and, since the number $k$ of steps is usually relatively small in practice (usually $k\ll n=|U|$ and we assume, in what follows, that $k<n$), Wang and Li \cite{WaLi10} introduced its parameterisation\footnote{We use terminology of the recent monograph \cite{DoFe13} on parameterised algorithms and complexity.}  by $k$. Algorithms for this parameterised problem were also studied in \cite{FAW2014,CoCrGaGuJo13,JOCO2014,CrGuYe13}. While in general the WSP is W[1]-hard \cite{WaLi10}, 
the WSP restricted\footnote{While we consider special families of constraints, we do not restrict authorisation lists.} to some practically important families of constraints is fixed-parameter tractable (FPT) ~\cite{CoCrGaGuJo13,CrGuYe13,WaLi10}. (Recall that a problem parameterised by $k$ is FPT if it can be solved by an FPT algorithm, i.e.\ an algorithm of running time $O^*(f(k))$, where $f$ is an arbitrary function depending on $k$ only, and $O^*$ suppresses not only constants, but also polynomial factors in $k$ and other parameters of the problem formulation.)

Many business rules are not concerned with the identities of the users that perform a set of steps.  Accordingly, we say a constraint $c = (T,\Theta)$ is \emph{user-independent} if, whenever $\theta \in \Theta$ and $\phi: U \rightarrow U$ is a permutation, then $\phi \circ \theta \in \Theta$. 
In other words, given a complete plan $\pi$ that satisfies $c$ and any permutation $\phi: U \rightarrow U$, the plan $\pi' : S \rightarrow U$, where $\pi'(s) = \phi(\pi(s))$, also satisfies $c$.
The class of user-independent constraints is general enough in many practical cases; for example, all the constraints defined in the ANSI RBAC standard \cite{ANSI04} are user-independent. 
Most of the constraints studied in \cite{FAW2014,JOCO2014,CrGuYe13,WaLi10} and other papers are also user-independent. 
Classical examples of user-independent constraints are the requirements that two steps are performed by either two different users ({\em separation-of-duty}), or the same user ({\em binding-of-duty}). More complex constraints state that at least/at most/exactly $r$ users are required to complete some sensitive set of steps (these constraints belong to the family of \emph{counting} constraints), where $r$ is usually small. 
A simple reduction from {\sc Graph Colouring} shows that the WSP restricted to the separation-of-duty constraints is already NP-hard \cite{WaLi10}.

The WSP is an important applied problem and is thoroughly studied
in the literature.  However, as was shown by Cohen et al.~\cite{FAW2014}, 
the methods developed so far  were 
capable of solving user-independent WSP instances only for relatively small values of $k$.  
%Effectively, there was no method to tackle such WSP instances even with moderate $k$. 
In this paper we propose a new approach that, compared to the existing solution methods, significantly extends the number of steps in practically solvable instances now covering the values of $k$ expected in the majority of real-world instances. 
Importantly, the proposed method is relatively simple to implement or extend with new constraints, such that its accessibility is similar to that of SAT-solvers used by practitioners~\cite{WaLi10}.

%In this paper, we propose a new algorithm to solve the WSP with user-independent constraints. 

The proposed solution method is a deterministic algorithm that uses backtracking to browse the space of all the equivalence classes of partial solutions. 
We show that it is possible to test efficiently if there exists an authorised plan in a given equivalence class.
This makes our algorithm FPT as the number of equivalence classes is exponential in $k$ only.
%We show that the time complexity of this approach is superior to the previous results.
% experimental results for the new FPT algorithm
%We experimentally compare the new algorithm to the two state-of-the-art methods~\cite{FAW2014,JOCO2014} and 
%show that the new algorithm clearly outperforms both solvers on a wide range of test instances and significantly 
%extends the domain of practically solvable instances of the WSP.

%%%%%%%%%%%%%%%%%%%%%%%%%
% paper organization
%In Section~\ref{sec:related_notations}, we briefly describe related work and 
%introduce some notations and basic concepts. Section~\ref{sec:backtracking} describes details of the new backtracking FPT algorithm searching the space of patterns in a depth-first order, a procedure used to obtain a plan realising a given pattern, and the corresponding algorithm to search for matchings in a bipartite graph. In Section~\ref{sec:comparison}, we compare the two FPT approaches, i.e. our new FPT algorithm with that provided in \cite{FAW2014,CoCrGaGuJo13,JOCO2014}. Section~\ref{sec:experiments} describes computational experiments in support of the new FPT approach to solve WSP with counting constraints. Finally, Section~\ref{sec:conclusion} provides some general remarks and conclusions.

%%%%%%%%%%%%%%%%%%%%%%%%%%%
\section{Patterns and the User-Iterative Algorithm}
\label{sec:patterns-and-ui}
% related work; connection to the previous work
% generic FPT algorithm based on encoding equivalence classes into patterns; theory vs practice 

%Pattern is a key concept of the algorithm proposed in this paper and the algorithm in~\cite{CoCrGaGuJo13}.

A \emph{plan} is a function $\pi: T \rightarrow U$, where $T \subseteq S$ (note that if $T = S$ then $\pi$ is a complete plan).  %A partial plan is authorised if ...
%\begin{itemize}
%\item for all $s \in S$, $s \in A(\pi(s))$ (each step is allocated to an authorised user);
%\item for all $(T,\Theta) \in C$, $\pi|_{T} = \theta$ for some $\theta \in \Theta$ (every constraint is satisfied).
%\end{itemize}
We define an \emph{equivalence relation} on the set of all plans, which is a special case of an equivalence relation defined in \cite{CoCrGaGuJo13}.
For user-independent constraints, two plans $\pi : T \rightarrow U$ and $\pi' : T' \rightarrow U$ are \emph{equivalent}, denoted by $\pi \approx \pi'$, if and only if $T = T'$, and $\pi(s) = \pi(t)$ if and only if $\pi'(s) = \pi'(t)$ for every $s, t \in T$.
Assuming an ordering $s_1, s_2, \dots, s_k$ of steps $S$,
every plan $\pi : T \rightarrow U$ can be encoded into a \emph{pattern} $P = P(\pi) = (x_1,\dots,x_k)$ defined by:
\begin{equation}
 x_i =
  \begin{cases}
    0 & \text{if } s_i \notin T, \\
    1 & \text{if } i = 1 \text{ and } s_1 \in T,\\
    x_j & \text{if } \pi(s_i) = \pi(s_j) \text{ and } j < i, \\
    \max\set{x_1, x_2, \dots, x_{i-1}} + 1 & \text{otherwise.}
  \end{cases}
\end{equation}
The pattern $P(\pi)$ uniquely encodes the equivalence class of $\pi$, and $P(\pi) = P(\pi')$ for every $\pi'$ in that equivalence class~\cite{CoCrGaGuJo13}. %\todo{(Is there a better way of saying it?)}.
The pattern $P$ represents an assignment of steps in $T$ to some users in any plan of the equivalence class of $\pi$.  
We say that a pattern is \emph{complete} if $x_i \neq 0$ for $i = 1, 2, \ldots, k$.

The state-of-the-art FPT algorithm for the WSP with counting constraints proposed in~\cite{FAW2014} and called here User-Iterative (UI), iterates over the set of users and gradually computes all encoded equivalence classes of valid plans until it finds a complete solution to the problem, or all the users have been considered. 
Effectively, it uses the breadth-first search in the space of plans.  
In the breadth-first search tree, equivalent plans can be generated but they are detected efficiently using patterns and the corresponding search branches are then merged together.  
Since the UI algorithm generates a polynomial number of plans per equivalence class and the number of equivalence classes is exponential in $k$ only, the UI algorithm is FPT.
The results of~\cite{FAW2014,JOCO2014} show that the generic user-iterative FPT algorithm of~\cite{CoCrGaGuJo13} has a practical value, and its implementations are able to outperform the well-known pseudo-Boolean SAT solver SAT4J~\cite{BePa10}.

In this paper we propose a new FPT solution method for the WSP which also exploits equivalence classes and patterns but in a more efficient manner.  
Among other advantages, our algorithm never generates multiple plans within the same equivalence class.  
For further comparison of our algorithm with the UI algorithm, see Section~\ref{sec:comparison}.

\section{The Pattern-Backtracking Algorithm}
\label{sec:pb}

We call our new method \emph{Pattern-Backtracking} (PB) as it uses the backtracking approach to browse the search space of patterns.
%In this section we describe the new FPT algorithm, which we call \emph{Pattern-Backtracking} (PB), that uses a backtracking approach to browse the search space of patterns.
To describe it, we introduce several additional notations.
We will say that a plan $\pi : T \rightarrow U$ is \emph{authorised} if $s \in A(\pi(s))$ for every $s \in T$, \emph{eligible} if it does not violate any constraint in $C$, and \emph{valid} if it is both authorised and eligible.
Similarly, a pattern $P$ is \emph{authorised}, \emph{eligible} or \emph{valid} if there exists a plan $\pi$ such that $P(\pi) = P$ and $\pi$ is authorised, eligible or valid, respectively.
By $P(s_i)$ we denote the value $x_i$ in $P = (x_1, x_2, \ldots, x_k)$.  
We also use notations $A^{-1}(s) = \{ u \in U :\; s \in A(u) \}$ for the set of users authorised for step $s \in S$ and $P^{-1}(x_i) = \{ s \in S :\; P(s) = x_i \}$ for all the steps assigned to the same user encoded by the value of $x_i$ in $P$\@.
%, and $P^{-1}(0) = S\setminus T$ in partial patterns.
Note that $P^{-1}(x_i) \neq \emptyset$ for $i = 1, 2, \ldots, k$ for any complete pattern.

%Unlike the UI algorithm, PB searches for eligible patterns first, and only then tries to convert them into valid plans.  
%Hence, less effort is put into exploring search-space branches with no eligible patterns.
%The other advantage of the backtracking algorithm is that it explores the search tree in the depth-first order.  This gives the algorithm flexibility to change the order of steps in which they are assigned and reduces memory consumption.

In Section~\ref{sec:pattern-validity} we show how to find a valid plan for an eligible pattern, which is an essential part of our algorithm, and in Section~\ref{sec:pb-backtracking} we describe the algorithm itself.

%%%%%%%%%%%%%%%%%%%%
\subsection{Pattern Validity Test}
\label{sec:pattern-validity}

The PB algorithm searches the space of patterns; once an eligible complete pattern $P$ is found, we need to check if it is valid and, if it is, then to find a plan $\pi$ such that $P = P(\pi)$.
The following theorem allows us to address these two questions efficiently.

For a complete pattern $P = (x_1, x_2, \ldots, x_k)$, let $X = \{ x_i :\; i = 1, 2, \ldots, k \}$ (note that the cardinality of the set $X$ may be smaller than $k$).
Let $G = (X \cup U, E)$ be a bipartite graph, where $(x_i, u) \in E$ if and only if $u \in A^{-1}(s)$ for each $s \in P^{-1}(x_i)$, $x_i\in X$. 

%In this section, we deal with two problems:\\ 
%(a) Given an eligible pattern $P$, determine whether it is valid (i.e. authorised);\\ 
%(b) Given a valid pattern $P$, provide a valid plan $\pi$ realising the pattern $P$.\\

%%%%%%%%%%%%%%%%%%%%%%
\begin{theorem}
\label{th:matching}
 A pattern $P$ is authorised if and only if $G$ has a matching of size $|X|$.
\end{theorem}
\begin{proof}
Suppose $M$ is a matching of size $|X|$ in $G$. 
Construct a plan $\pi$ as follows: for each edge $(x_i, u) \in M$ and $s \in P^{-1}(x_i)$, set $\pi(s) = u$.  
Since $M$ covers $x_i$ for every $i = 1, 2, \ldots, k$ and $P$ is a complete pattern, the above procedure defines $\pi(s)$ for every step $s \in S$.
Hence, $\pi$ is a complete plan.
Now observe that, for each $x_i \in X$, all the steps $P^{-1}(x_i)$ are assigned to exactly one user, and if $x_i \neq x_j$ for some $i, j \in \{ 1, 2, \ldots k \}$, then $\pi(s_i) \neq \pi(s_j)$ by definition of matching.
Therefore $P(\pi) = P$. 
Observe also that $\pi$ respects the authorisation lists; for each edge $(x_i, u) \in M \subseteq E$ and each step $s \in P^{-1}(x_i)$, we guarantee that $u \in A(s)$. 
Thus, plan $\pi$ is authorised and, hence, pattern $P = P(\pi)$ is also authorised.

\smallskip
%%%%
On the other hand, assume there exists an authorised plan $\pi$ such that $P(\pi) = P$ for a given pattern $P$\@.  
Let $X = \{ x_i :\; i = 1, 2, \ldots, k \}$.  
Construct a set $M$ as follows: $M = \{ (x_i, u) :\; x_i \in X \text{ and } \exists s \in P^{-1}(x_i) \text{ s.t.\ } u = \pi(s) \}$.  
Consider a pair $(x_i, u) \in M$, and find some $s \in P^{-1}(x_i)$.
Note that, as $P = P(\pi)$ and by definition of pattern, $\pi(s') = u$ for every $s' \in P^{-1}(x_i)$.  
Since $\pi$ is authorised, $u \in A(s')$ for every $s' \in P^{-1}(x_i)$, i.e.\ $(x_i, u) \in E$ and $M \subseteq E$.  In other words, $M$ is a subset of edges of $G$.

Now notice that, for each $x_i \in X$, there exists at most one edge $(x_i, u) \in M$ as $\pi(s') = u$ for every $s' \in P^{-1}(x_i)$.  Moreover, for each $u \in U$, there exists at most one edge $(x_i, u) \in M$ as otherwise there would exist some $i, j \in \{ 1, 2, \ldots, k \}$ such that $\pi(s_i) = \pi(s_j)$ and $x_i \neq x_j$, which violates $P = P(\pi)$.  Hence, the edge set $M$ is disjoint.  Finally, $|M| = |X|$ because $P^{-1}(x_i)$ is non-empty for every $x_i \in X$.  We conclude that $M$ is a matching in $G$ of size $|X|$.
\qed
\end{proof}

%%%%%%%%%%%%%%%%%%%%%%
Theorem~\ref{th:matching} implies that, to determine whether an eligible pattern $P$ is valid, it is enough to construct the bipartite graph $G = (X \cup U, E)$ and to find a maximum size matching in $G$.  
It also provides an algorithm for converting a maximum matching $M$ of size $|X|$ in $G$ into a valid plan $\pi$ such that $P(\pi) = P$.

The matching problem arising in Theorem~\ref{th:matching} has some interesting properties:
\begin{itemize}
	\item The bipartite graph $G = (X \cup U, E)$ is highly unbalanced as $|X| \le k$, and we assume that $|U| \gg k$.  
	It is easy to see that the maximum length of an augmenting path in $G$ is $|X| \le k$ and, hence, the time complexity of the Hungarian and Hopcroft-Karp methods are $O(k^3)$ and $O(k^{2.5})$, respectively.
	\item We are interested only in matchings of size $|X|$.  If the maximum matching is of a smaller size, we do not need to retrieve it.
	\item Once a matching of size $|X|$ is found, the PB algorithm terminates since a valid plan is found.  
	However, the algorithm might test an exponential number (in $k$) of graphs with the maximum matching of size smaller than $|X|$.
	Hence, we are mainly interested in time of checking whether the maximum matching is smaller than $|X|$.
\end{itemize}

To exploit the above features, we use the Hungarian method with a speed-up heuristic provided by the following proposition.
%Moreover, we can implement several speed-up heuristics to detect non-existence of a matching of size $|X|$:\\
%%\begin{enumerate}
%--- If $\nexists (x, u) \in E$ for some $x \in X$ ($x\in X$ is an isolate in $G$?), then there is no matching of size $|X|$ in $G$.\\
%--- If $|\{ (x, u) \in E :\; u \in U \}| \ge |X|$, one can remove vertex $x$ from $G$ without affecting a correct outcome of the algorithm.
%%\end{enumerate}
%
%The following lemma provides an observation used as a heuristic speed-up in our algorithm to detect instances of the graph $G$ having the maximum matching size less than $|X|$ earlier.
\begin{proposition}
\label{lemma:matching}
If $M = \{ (x_{\chi(1)}, u_{\psi(1)}), \ldots, (x_{\chi(t)}, u_{\psi(t)}) \}$, $t < |X|$, is a matching in the graph $G = (X \cup U, E)$ such that there exists no $M$-augmenting path in $G$ starting at a vertex $x_{\chi(t+1)} \in X$, then there is no matching covering all vertices of $X$ in $G$.
\end{proposition}
\begin{proof}
W.l.o.g, assume that $M=\{(x_1,u_1),\dots ,(x_t,u_t)\}$ and $x_{\chi(t+1)}=x_{t+1}$. Now suppose that $G$ has a matching $$M' = \{ (x_1, y_1), (x_2, y_2),\ldots, (x_k, y_k) \},\ y_i \in U,\ i = 1, 2, \ldots, k$$ covering all of $X$.
Consider the symmetric difference of two matchings $H = (M \cup M') \setminus (M \cap M')$. 
Since every vertex of $H$ has degree at most $2$, the graph induced by $H$ consists of some disjoint paths and even cycles having edges alternating between $M$ and $M'$. 
Since $x_{t+1}$ is not in $M$ but covered by $M'$, it is an end point of one of the alternating paths in $H$, say $P_{t+1}$. 
Now, it is possible to see that $P_{t+1}$ is an augmenting path in $G$ with respect to the matching $M$; since, starting at $x_{t+1}$, every time we use an edge of $M'$ to go from a vertex in $X$ to a vertex in $U$, $P_{t+1}$ must have an end point in a vertex of $U$ not covered by $M$ ($M'$ covers all the vertices in $X$). 
This contradicts the fact that there is no augmenting path in $G$ starting at $x_{t+1}$ with respect to $M$.\squareforqed
\end{proof}

The result of Proposition~\ref{lemma:matching} allows us to terminate the Hungarian algorithm as soon as a vertex $x_i \in X$ is found such that no augmenting path starting from $x_i$ can be obtained.
Construction of the graph takes $O(kn)$ time, and solving the maximum matching problem in $G$ with the Hungarian method takes $O(k^3)$ time.

%%%%%%%%%%%%%%%%%%%%%%%%%%%
\subsection{The Backtracking Algorithm}
\label{sec:pb-backtracking}

The PB algorithm uses a backtracking technique to search the space of patterns, and for each eligible pattern, it verifies whether such a pattern is valid.  
If a valid pattern $P$ is found, the algorithm returns a complete plan $\pi$ such that $P(\pi) = P$, see Section~\ref{sec:pattern-validity} for details.  
If no valid pattern is found during the search, the instance is unsatisfiable.

The calling procedure for the PB algorithm is shown in Algorithm~\ref{alg:start}, which in turn calls the recursive search function in Algorithm~\ref{alg:recursion}. 
%The branching heuristic is given in Algorithm~\ref{alg:step-selection}. 
%Notice that the initialisation of patterns (the staring point) in Algorithm~\ref{alg:start} is similar to the initialization of the BFS-like algorithms of \cite{FAW2014,CoCrGaGuJo13,JOCO2014}. 
%Algorithm~\ref{alg:recursion} is a recursive call of the backtracking mechanism.
The recursive function tries all possible extensions $P'$ of the current pattern $P$ by adding one new step $s$ to it (line~\ref{line:extend-pattern}).
The step $s$ is selected heuristically (line~\ref{line:select-s}), where function $\rho(s)$ is an empirically tuned function indicating the importance of step $s$ in narrowing down the search space.  
The implementation of $\rho(s)$ depends on the specific types of constraints involved in the instance and should reflect the intuition regarding the structure of the problem.
See (\ref{eq:rho}) in Section~\ref{sec:experiments} for a particular implementation of $\rho(s)$ for the types of constraints we used in our computational study.
Note that our branching heuristic dynamically changes the steps ordering used in the pattern definition in Section~\ref{sec:patterns-and-ui}.
Nevertheless, this does not affect any theoretical properties of the pattern.

We use a heuristic (necessary but not sufficient) test (lines~\ref{line:authorisation-heuristic-begin}--\ref{line:authorisation-heuristic-end} of Algorithm~\ref{alg:recursion}) to check whether the pattern $P'$ can be authorised; that allows us to prune branches which are easily provable to include no authorised patterns.

In line~\ref{line:pattern-eligibility-test}, the algorithm checks whether the new pattern $P'$ violates any constraints and, if not, then executes the recursive call.

\SetKwInOut{Input}{input}
\SetKwInOut{Output}{output}

%%%%%%%%%%%%%%%%%%
\begin{algorithm}[tb]
\caption{Backtracking search initialisation (entry procedure of PB)\label{alg:start}}
\Input {WSP instance $W = (S, U, \mathcal{A}, C)$}
\Output {Valid plan $\pi$ or UNSAT}

Initialise $P(s) \gets 0$ for each $s \in S$\;
$\pi \gets \text{Recursion}(P)$\;

\Return{$\pi$} ($\pi$ may be UNSAT here)\;
\end{algorithm}

%%%%%%%%%%%%%%
{\small
\begin{algorithm}[tb]
\caption{Recursion$(P)$ (recursive function for backtracking search)\label{alg:recursion}}
\Input {Pattern $P$}
\Output {Eligible plan or UNSAT if no eligible plan exists in this branch of the search tree}
Initialise the set $S' \subseteq S$ of assigned steps $S' \gets \{ s \in S :\; P(s) \neq 0 \}$\;
\If {$S' = S$}
{
	Verify if pattern $P$ is valid (using the matching algorithm of Theorem~\ref{th:matching})\;
	\If {pattern $P$ is valid}
	{
		\Return {plan $\pi$ realising $P$}\;
	}
	\Else
	{
		\Return {UNSAT}\;
	}
}
\Else
{
	Select unassigned step $s \in S \setminus S'$ that maximises $\rho(s)$\; \label{line:select-s}
	Calculate $k' \gets 1 + \max_{s \in S} P(s)$\;
	\For {$x = 1, 2, \ldots, k'$}
	{
		Set $P(s) \gets x$ to obtain a new pattern $P'$\; \label{line:extend-pattern}
		Compute the set of steps $Q$ assigned to $x$: $Q = \{ t \in S :\; P'(t) = x \}$\; \label{line:authorisation-heuristic-begin}
		\If {$|\bigcap_{t \in Q} A^{-1}(t)| = 0$}
		{
			Proceed to the next value of $x$ (reject $P'$)\; \label{line:authorisation-heuristic-end}
		}
		
		\If {$P'$ is an eligible pattern}
		{\label{line:pattern-eligibility-test}
			$\pi \gets \text{Recursion}(P')$\;
			\If {$\pi \neq UNSAT$}
			{
				\Return{$\pi$}\;
			}
		}
	}
}

\Return{UNSAT (for a particular branch of recursion; does not mean that the whole instance is unsat)}\;
\end{algorithm}
}

%%%%%%%%%%%%%%
%\begin{algorithm}[tb]
%\caption{Heuristic choice of a step for branching in the backtracking procedure\label{alg:step-selection}}
%
%\Input {Partial pattern $P$ with at least one unassigned step; if step $s$ is unassigned, $P(s) = 0$;}
%\Output {Step $s \in S$ such that $P(s) = 0$;}
%Initialise the set $S' \subset S$ of unassigned steps $S' = \{ s \in S :\; P(s) = 0 \}$\;
%Let $\rho : S' \rightarrow \mathbb{N}^0$ be a ranking function\;
%Initialise $\rho(s) = 0$ for every $s \in S'$\;
%\ForEach {$s \in S'$}
%{
%	%Initialise $\rho(s) \gets 0$\;
%	\ForEach {not-equals constraint $\{ t_1, t_2 \}$}
%	{
%		\If {$|\{ t_1, t_2 \} \cap S'| = 1$}
%		{
%			$\rho(s) \gets \rho(s) + 1$\;
%		}
%	}
%
%	\ForEach {at-most-$r$ constraint $T = \{t_1,t_2, \ldots,t_q\}$}
%	{
%		Compute the set $T'$ of already assigned steps $T' = \{ t \in T :\; P(t) > 0 \}$\;
%		Compute the set $X' \subseteq X$ of distinct assigned users $X' = \{ P(t) :\; t \in T' \}$\;
%
%		\If {$r = |X'|$}
%		{
%			$\rho(s) \gets \rho(s) + 100$\;
%		}
%		\ElseIf {$r - |X'| = 1$}
%		{
%			$\rho(s) \gets \rho(s) + 2$\;
%		}
%		\ElseIf {$r - |X'| = 2$}
%		{
%			$\rho(s) \gets \rho(s) + 1$\;
%		}
%	}
%}
%\Return {$s \in S'$ that maximises $\rho(s)$}\;
%\end{algorithm}

%%%%%%%%%%%%%
\section{Comparison of the PB and UI Algorithms} 
%{Comparison of the Two FPT Algorithmic Approaches}
\label{sec:comparison}

In this section we analyse the time and memory complexity of the PB algorithm and compare it to the UI algorithm.
 
%each of these constraints involves only a set of steps of cardinality bounded by a constant. 

Observe that each internal node (corresponding to an incomplete plan) in the search tree of the PB algorithm has at least two children, and each leaf in this tree corresponds to a complete pattern.
Thus, the total number of patterns considered by the PB algorithm is less than twice the number of complete patterns. 
Observe that the number of complete patterns equals the number of partitions of a set of size $k$, i.e.\ the $k$th Bell number $B_k$. 
Finally, observe that the PB algorithm spends time polynomial in $n$ on each node of the search tree.\footnote{Assuming that the WSP instance does not include any exotic constraints.}
Thus, the time complexity of the PB algorithm is $O^*(B_k)$.
The PB algorithm follows the depth-first search order and, hence, stores only one pattern at a time.  
At each leaf node, it also solves the matching problem generating a graph with $O(kn)$ edges.  
Hence, the memory complexity of the algorithm is $O(kn)$.

It is interesting to compare the PB algorithm to the UI algorithm (briefly described in Section~\ref{sec:patterns-and-ui}).
Despite both algorithms using the idea of equivalence classes and being FPT, they have very different working principles and properties.

%\begin{enumerate}
%	\item The time complexity of the UI algorithm is $O^*(3^kB_k\log B_k)$ \cite{CoCrGaGuJo13}, which is in %$O^*(3^kB_k)$ as $B_k=2^{k\log k(1-o(1))}$ \cite{BeTa10}. This indicates that the PB algorithm
%is asymptotically faster than the UI algorithm with respect to the
%parameter $k$.\footnote{It is easy to verify that the PB algorithm is
%also asymptotically faster than the UI algorithm with respect to
%$n$.}		\item 

\begin{enumerate}
	\item Observe that, in the worst case, the UI algorithm may store all patterns, and the number of patterns is $B_{k+1}$. Indeed, consider a pattern $P=(x_1,\dots ,x_k)$ and a set $\{s_1,\dots ,s_k,s_{k+1}\}.$ Then each partition of the set corresponds to a pattern of $P$, where $x_i=0$ if and only if $s_i$ and $s_{k+1}$ are in the same subset of the partition.
	Therefore,  the UI algorithm takes $O(k B_{k+1})$ memory, which is in sharp contrast to the PB algorithm that requires very little memory.
	Considering that, e.g.\ $B_{20} = 51\,724\,158\,235\,372$, memory consumption poses a serious bottleneck for the UI algorithm as the RAM capacity of any mainstream machine is well below the 
	value of $B_{20}$.
	Moreover, the UI algorithm accesses a large volume of data in a non-sequential order, which might have a dramatic effect on the algorithm's performance when implemented on a real machine as shown in~\cite{KaGuGo}.
%\item 

	\item From the practical point of view, the PB algorithm considers less
patterns than the UI algorithm ($O(B_k)$ vs.\ $O(B_{k+1})$) as
the PB algorithm assigns the steps in a strict order, avoiding generation of duplicate patterns.  Moreover, the
PB algorithm generates each pattern at most once, while the UI
algorithm is likely to generate a pattern several times rejecting the duplicates afterwards.  
%This indicates that the
%above estimation of the theoretical time complexity correlates with
%practical efficiency considerations.
%	\item

	\item Both algorithms use heuristics to determine the order in which the search tree is explored.
	However, while the UI algorithm has to use a certain fixed order of users for all the search branches, the PB algorithm has the flexibility of changing the order of steps in each branch of the search.
	Note that the order of assignments is crucial to the algorithm's performance as it can help to prune branches early.
\end{enumerate}

\section{Computational Experiments}
\label{sec:experiments}

In this section we empirically verify the efficiency of the PB algorithm.  We compare the following WSP solvers:
\begin{description}
	\item[PB] The algorithm proposed in this paper;
	\item[UI] Another FPT algorithm proposed in~\cite{CoCrGaGuJo13} and evaluated in \cite{FAW2014,JOCO2014};
	\item[SAT4J] A pseudo-Boolean SAT formulation~\cite{FAW2014,JOCO2014} of the problem solved with SAT4J.
\end{description}

Due to the difficulty of acquiring real-world WSP instances~\cite{FAW2014,WaLi10}, we use the random instance generator described in~\cite{FAW2014}.
Three families of user-independent constraints are used: \emph{not-equals} (also called \emph{separation-of-duty}) constraints \mbox{$(s, t, \neq)$}, \emph{at-most-$r$} constraints $(r, Q, \leqslant)$ and \emph{at-least-$r$} constraints $(r, Q, \geqslant)$.
A not-equals constraint $(s, t, \neq)$ is satisfied by a complete plan $\pi$ if and only if $\pi(s) \neq \pi(t)$.
An at-most-$r$ constraint $(r, Q, \leqslant)$ is satisfied if and only if $|\pi(Q)| \le r$, where $Q$ is the scope of the constraint.
Similarly, an at-least-$r$ constraint $(r, Q, \geqslant)$ is satisfied if and only if $|\pi(Q)| \ge r$.
We do not explicitly consider the widely used binding-of-duty constraints, that require two steps to be assigned to one user, as those can be trivially eliminated during preprocessing.
%former may be represented as a tuple $(r,Q,\leqslant)$, $Q \subseteq S$, $1\leqslant r \leqslant |Q|$, and is satisfied by any plan that allocates no more than $r$ users to the steps in $Q$.
%The separation-of-duty constraints $(s,t,\neq)$ are also called \emph{not-equals constraints}; a plan $\pi$ satisfies the constraint $(s,t,\neq)$ if $\pi(s) \neq \pi(t)$. 
%The counting constraints that we consider are \emph{at-most-$r$} and \emph{at-least-$r$} constraints.
%The former may be represented as a tuple $(r,Q,\leqslant)$, $Q \subseteq S$, $1\leqslant r \leqslant |Q|$, and is satisfied by any plan that allocates no more than $r$ users to the steps in $Q$.
%The latter may be represented as $(r,Q,\geqslant)$ and is satisfied by any plan that allocates at least $r$ users to the steps in $Q$, and we assume that $|Q|$ is bounded by a constant.
%From a practical point of view, these constraints naturally generalise binding-of-duty and separation-of-duty constraints: a binding-of-duty constraint $(s',t',=)$ can be thought of as an at-most-$1$ constraint $(1,\{s',t'\},%\leqslant)$, and a separation-of-duty constraint $(s,t,\neq)$ can be thought of as an at-least-$2$ constraint $(2,\{s,t\},\geqslant)$.
While the binding-of-duty and separation-of-duty constraints provide the basic modelling capabilities, the at-most-$r$ and at-least-$r$ constraints impose more general ``confidentiality'' and ``diversity'' requirements on the workflow, which can be important in some business environments.

The instance generator (available for downloading~\cite{SourceCodes}) takes four parameters: the number of steps $k$, the number of not-equals constraints $e$, the number of at-most and at-least constraints $c$ and the random generator seed value.
Each instance has $n = 10k$ users.  
For each user $u \in U$, it generates a uniformly random authorisation list $A(u)$ such that $|A(u)|$ is selected uniformly from $\{ 1, 2, \ldots, \lceil 0.5 k \rceil \}$ at random.  
It also generates $e$ distinct not-equals, $c$ at-most and $c$ at-least constraints uniformly at random.
All at-most and at-least constraints are of the form $(3, Q, \sigma)$, where $|Q| = 5$ and $\sigma \in \{ \leqslant, \geqslant \}$.

Our test machine is based on two Intel Xeon CPU E5-2630 v2 (2.6~GHz) and has 32~GB RAM installed.  
Hyper-threading is enabled, but we never run more than one experiment per physical CPU core concurrently.  
The PB algorithm is implemented in C\#, and the UI algorithm is implemented in C++.  
Concurrency is not exploited in any of the tested solution methods.
The PB algorithm is also available for downloading~\cite{SourceCodes}.

\bigskip

The branching heuristic implemented in line~\ref{line:select-s} of Algorithm~\ref{alg:recursion} selects a step $s \in S$ that maximises a ranking function $\rho(s)$:
\begin{equation}
\label{eq:rho}
\rho(s) = c_{\neq}(P) + \alpha c^0_\leq(P) + \beta c^1_\leq(P) + \gamma c^2_\leq(P),
\end{equation}
where $c_{\neq}(P)$ is the number of not-equals constraints involving step $s$, $c^i_\leq(P)$ is the number of at-most-$r$ constraints involving $s$ such that $r - i$ distinct users are already assigned to it, and $\alpha$, $\beta$ and $\gamma$ are parameters.  The intuition is that the steps $s$ that maximise $\rho(s)$ are tightening the search space quickly.  The parameters $\alpha$, $\beta$ and $\gamma$ were selected empirically.  We found out that the algorithm is not very sensitive to the values of these parameters, and settled down at $\alpha = 100$, $\beta = 2$ and $\gamma = 1$.  Note that the function does not account for at-least constraints.  This reflects our empirical observation that the at-least constraints are usually relatively weak in our instances and rarely help in pruning branches of search.

\bigskip

We started from establishing what parameter values make the instances hard.  
However, due to the lack of space, we provide only the conclusions drawn from this series of experiments.
As it could be expected, greatly under- and over-subscribed instances are easier to solve than the instances in the region between those two extremes.
The behaviour of the analysed solvers is consistent in this regard.
The particular values of the number of not-equals constraints $e$ and the number $c$ of at-most and at-least constraints that make the instances most challenging depend on $k$.
Thus, in our final experiment, which is to establish the maximum size $k$ of instances practically solvable by each of the methods, we considered several instances with a range of parameters to ensure that at least one of them is hard.
In particular, we fixed the \emph{density} of not-equals constraints, calculated as $d = \frac{2e}{k (k - 1)} \cdot 100\%$, at $d = 10\%$ and the number $c$ of at-most and at-least constraints at each of $c = 1.0k$, $c = 1.2k$ and $c = 1.4k$, producing three instances for each $k$ and seed value.

%
%At first, we fix $k = c = 20$ and test the running time of each algorithm on instances with various densities $d$ of not-equals constraints, where the density is calculated as $d = \frac{2e}{k (k - 1)} \cdot 100\%$.  
%The results are presented in Figure~\ref{fig:change-density}.  
%As expected, the probability for an instance to be unsatisfiable grows with the not-equals density (see the ``UNSAT'' plot).  
%Observe that, for the SAT and PB solvers, the hardest instances are the phase-transition instances in between under-subscribed instances on the left (all are satisfiable) and over-subscribed instances on the right (all are unsatisfiable).  
%The UI algorithm is relatively slow on lightly-constrained instances but clearly outperforms SAT on more constrained instances, which is consistent with the discussion in~\cite{FAW2014}.  
%The PB algorithm outperforms both SAT and UI solvers by three to six orders of magnitude on the whole range of instances.
%
%In our second experiment (Figure~\ref{fig:change-c}), we fixed $k = 20$, $d = 10\%$ and changed $c$ in a wide range.  
%In contrast to the above results, the running time of the solvers on the over-subscribed instances does not notably fall with the increase of $c$.
%This indicates that the at-most and at-least constraints generally provide less opportunities for early branch pruning.
%
%
%
\begin{figure}
\input{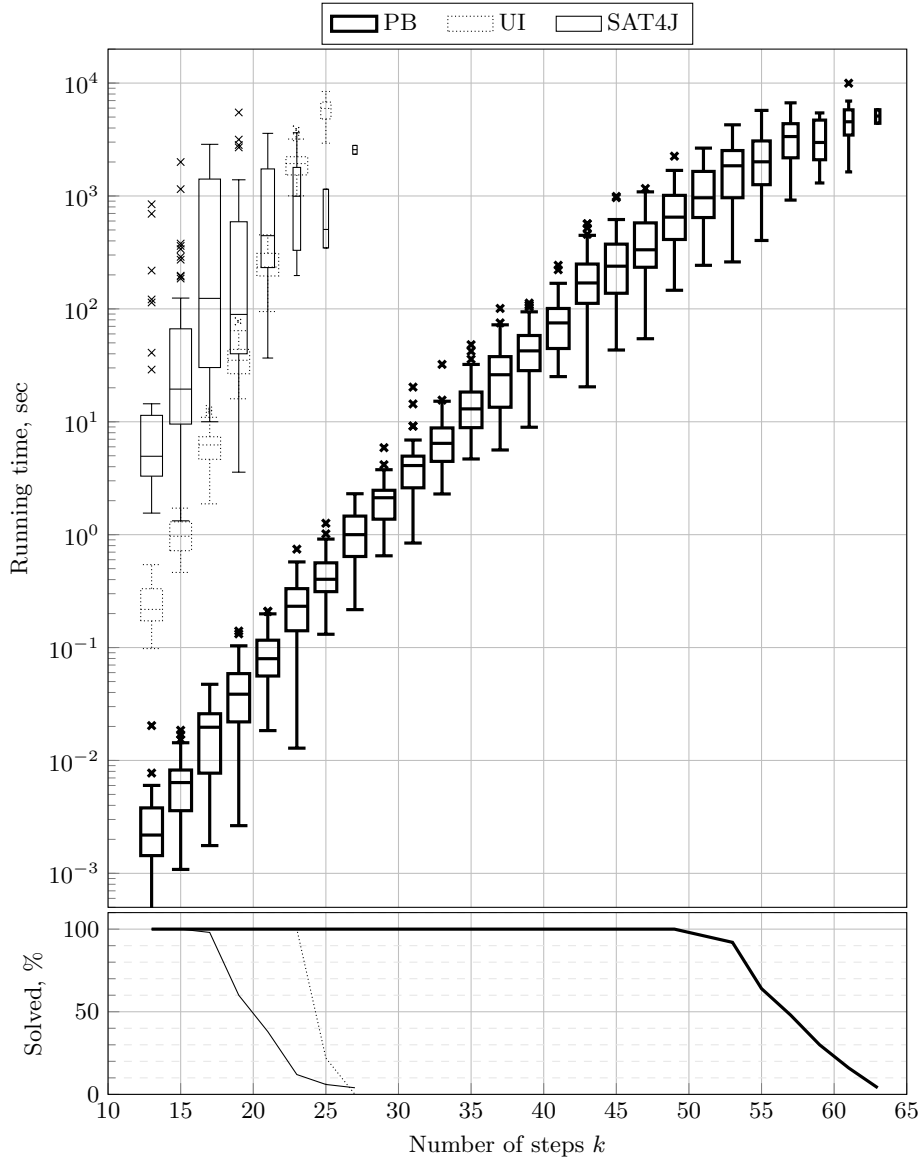}

\caption{Running time vs.\ number of steps $k$. For each $k$, we generate 50 instance sets with different seed values.  
The distributions are presented in the boxplot form, where the width of a box is proportional to the number of instance sets on which the solver succeeded.  
The plot at the bottom of the figure also shows the success rate of each solver.}
\label{fig:change-k}
\end{figure}
%
%
%
%In our final experiment, we solve sets of instances; for each $k$ and seed value, we generate four instances: $c = 0.8k$, $c = k$, $c = 1.2k$ and $c = 1.4k$.
%The not-equals density $d$ is fixed at $d = 10\%$ in this experiment.  
Each solver is given one hour limitation for each instance from the set. 
If a solver fails on at least one of the instances (could not terminate within 1~hour), we say that it fails on the whole set.
The intention is to make sure that the solver can tackle hard satisfiable and unsatisfiable instances within a reasonable time.
The results are presented in Figure~\ref{fig:change-k} in the form of boxplots.
The percentage of runs in which the solver succeeded is shown as the width of the box.
This information is also provided at the bottom of Figure~\ref{fig:change-k}.

The PB algorithm, being faster than the two other methods by several orders of magnitude, reliably solves all the instances of size up to $k = 49$.
Compare it to the UI and SAT4J solvers that succeed only for $k \le 23$ and $k \le 15$, respectively.
Moreover, its running time grows slower than that of the UI and SAT4J solvers, which indicates that it has higher potential if more computational power is allocated.
In other words, thanks to our new solution method, the previously unapproachable problem instances of practical sizes can now be routinely tackled.

%On small instance, the PB algorithm is faster than the two other algorithms by several orders of magnitude.
%The reader may also notice that the running time of each of the algorithms is approximately exponential in $k$, though the PB algorithm's running time grows slower than that of UI and SAT4J\@.
%(One may notice that the relatively straight lines of each of the solver curve down at the end; that is due to the experiment design, as the running time of each algorithm is limited.)
%As a result, it significantly extends the number of steps $k$ that can be tackled in practice; for example, given one hour time per instance, the previous state-of-the-art algorithm solves less than 25\% of the instance sets if $k = 25$ while the PB algorithm solves 100\% of the instances for every $k < 50$.
%In other words, previously unapproachable instances of practical sizes can now be easily tackled with the PB algorithm.
%Figure~\ref{fig:change-k} also demonstrates that the running time of the UI algorithm has relatively low dispersion while the SAT solver's dispersion is very high.
%As a result, the UI algorithm reliably solves all the instances of size up to $k = 23$ but then the number of solved instances sharply declines.

%%%%%%%%%%%%%
\section{Conclusion}
\label{sec:conclusion}

We proposed a new FPT algorithm for the WSP with user-independent constraints.  
%It is based on the idea of equivalence relation patterns, which was first exploited in the previous state-of-the-art algorithm.  
%Nevertheless, the new algorithm is superior to the old one due to the clear separation between eligibility and authorisation constraints and a more efficient search method.  
Our experimental analysis have shown that the new algorithm outperforms all the methods in the literature by several orders of magnitude and significantly extends the domain of practically solvable instances.
Another advantage of the new FPT algorithm is that it is relatively easy to implement and extend; for example, it is straightforward to parallelise it.

Future research is needed to establish further potential to improve the algorithm's performance.
Particular attention has to be paid to the branching heuristic.
Thorough empirical analysis has to be conducted to investigate the performance of the algorithms on easy and hard instances.

Another relevant subject was recently studied in~\cite{ValuedWSP-SACMAT}; the paper introduces an optimisation version of WSP and proposes an FPT branch and bound algorithm inspired by Pattern Backtracking.

\paragraph{Acknowledgment.}
This research was partially supported by EPSRC grant EP/\linebreak[1]K005162/1.
The source codes of the Pattern Backtracking algorithm and the instance generator are publicly available~\cite{SourceCodes}.

%\bibliography{refs}
%\bibliographystyle{splncs03}

\end{document}